\documentclass[submission]{eptcs}

\providecommand{\event}{GanDalF 2014} 

\usepackage[T1]{fontenc}
\usepackage{mathpazo}
\usepackage{times}
\usepackage{pta}

\title{Improved Undecidability Results for Reachability Games on Recursive Timed Automata} 
\author{
  Shankara Narayanan Krishna \quad Lakshmi Manasa \quad Ashutosh Trivedi
  \institute{Indian Institute of Technology Bombay, Mumbai, INDIA}
  \email{\{krishnas, manasa, trivedi\}@cse.iitb.ac.in}
}

\def\titlerunning{Games on Recursive Timed Automata}
\def\authorrunning{Krishna, Manasa, and Trivedi}

\begin{document}
\maketitle

\begin{abstract}
  {\bf Abstract}. We study reachability games on recursive timed automata (RTA)
  that generalize Alur-Dill timed automata with recursive procedure invocation
  mechanism similar to recursive state machines. 
  It is known that deciding the winner in reachability games on RTA is
  undecidable for automata with two or more clocks, while the 
  problem is decidable for automata with only one clock.
  Ouaknine and Worrell  recently proposed a time-bounded theory of real-time
  verification by claiming that restriction to bounded-time recovers
  decidability for several key decision problem related to real-time
  verification. 
  We revisited games on recursive timed automata with time-bounded
  restriction in the hope of recovering decidability. 
  However, we found that the problem still remains undecidable for  recursive
  timed automata with three or more clocks.   
  Using similar proof techniques we characterize a decidability frontier
  for a generalization of RTA to recursive stopwatch
  automata.  
\end{abstract}

\section{Introduction}
Timed automata, introduced by Alur and Dill~\cite{AD94}, extend finite state
machines with a finite set of continuous variables called clocks that grow with
uniform rate in each state. 
Since the syntax of timed automata allows guarding the transitions and states with
simple constraints on clocks and resetting clocks during a transition, timed
automata can express complex timing based properties of real-time systems. 
The seminal paper of timed automata~\cite{AD90} showed the decidability of the
fundamental reachability problem for the timed automata, 
that paved
the way for the success of timed automata as the specification and verification
formalism for real-time systems.  

\emph{Recursive timed automata} (RTAs)~\cite{TW10} extend timed automata with
recursion to model real-time software systems. 
Formally, an RTA is a finite collection of components where each component is a
timed automaton that in addition to making transitions between various states,
can have transitions to ``boxes'' that are mapped to other components modeling a
potentially recursive call to a component. 
During such invocation a limited information can be  passed through clock values 
from the ``caller'' component to the ``called'' component via two different
mechanism: a) \emph{pass-by-value}, where upon returning from the called
component a clock assumes the value prior to the invocation, and b)
\emph{pass-by-reference}, where upon return a clock reflects any changes to the
value inside the invoked procedure. 
The reachability problem for RTA is known~\cite{TW10} to be
undecidable for RTA with three or more clocks. 

In this paper we study reachability games on recursive timed automata that are
played between two players---called \ach and \tort---who take turns to move a
token along the infinite graph of configurations (context, states, and clock
valuations) of recursive timed automata. 
In a reachability game the goal of Player~1 is to reach a desirable set of
target states, while the goal of Player~2 is to avoid it.  
The reachability game problem is to decide the winner in a reachability game. 
\begin{figure}[t]
  \centering
  \makebox[2cm]{{\small
  \begin{tikzpicture}[node distance=4cm]

    \draw(-2, -1) rectangle (2.5,1);
    \draw (2.3, 1.2) node {$M_1$};
    
    \node[loc](u1) at (-2, 0.5) {$u_1$};
    \node[loc](u2) at (-2, -0.5) {$u_2$}; 

    \node[loc](u3) at (2.5, 0) {$u_3$}; 
    
    \node[boxloc](b1) at (0.5, 0) {$\:\;~b_1:M_2~\:\;$};
    \node[varpass] () [below of =b1, node distance=5mm]{$(x)$};

    \node[port](b1v1) at (-0.4, 0) {}; 
    \node[port](b1v2) at (1.4, 0) {}; 

    \draw[trans] (u1)--(b1v1)  node [midway, above]{$x {=} 1$};
    \draw[trans] (u2)--(b1v1)  node [midway, below]{$x {<} 1$};
    \draw[trans] (b1v2)--(u3) node [midway, above]{$x {=} 0$};
    
    \draw(3.5, -1) rectangle (7.5,1);
    \draw (7.3, 1.2) node {$M_2$};
    
    \node[oloc](v1) at (3.5, 0) {$v_1$};
    \node[loc](v2) at (7.5, 0) {$v_2$};
    
    \node[boxloc](c1) at (5.5, 0.5) {$~b_2:M_2~$};
    \node[port](c1v1) at (4.8, 0.5) {}; 
    \node[port](c1v2) at (6.2, 0.5) {}; 
    
    \draw[otrans] (v1)-- +(0.2, 0.5) --(c1v1) node [midway, above]{$x {=} 1$} ;
    \draw[trans] (c1v2)-- +(0.7, 0) -- (v2);
    \draw[otrans] (v1)-- + (0.2, -0.5) -- +(3.5, -0.5) node [midway,
    above]{$x {=} 1, \set{x}$} -- (v2);
    
  \end{tikzpicture}
}}
  \caption{Reachability game on a recursive timed automata with
    one clock and two components} 
  \label{fig:example2}
\end{figure}
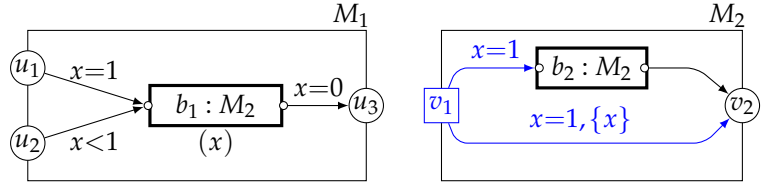
\begin{example}
  The visual presentation of a reachability game on recursive timed automaton with
  two components $M_1$ and $M_2$, and one clock variable $x$
  is shown in Figure~\ref{fig:example2} (inspired by example in \cite{TW10}) where  component $M_1$ calls component
  $M_2$ via box $b_1$ and component $M_2$ recursively calls itself via box $b_2$.
  Components are shown as thinly framed rectangles with their names 
  written next to upper right corner.
  Various control states, or ``nodes'', of the components are shown as circles
  or blue squares with their labels written inside them, e.g. see node $u_1$. 
  The circles are \ach states (see $u_1$) while blue squares ($v_1$) are  \tort states. 
  Entry nodes of a component appear on the left of the component (see $u_1$),
  while exit nodes appear on the right (see $u_3$). 
  Boxes are shown as thickly framed rectangles inside components
  labelled $b:M$, where $b$ is the label of the box, $M$ is the
  component it is mapped to.
  The tuple of clocks passed to $M$ by value, if any, are shown below the box,
  and the rest of the variables are passed by reference. 
  For example, in the figure clock $x$ is passed during the invocation of
  component $M_2$ via box $b_1$, while no clock is passed by value to component
  $M_2$ via box $b_1$. 
  Each transition is labelled with a guard and the set of reset variables,
  (e.g. transition from node $v_1$ to $v_2$ can be taken only
  when variable $x{<}1$, and after taking this transition, variable $x$ is
  reset). 
  To minimize clutter we omit empty reset sets.
\end{example}

Trivedi and Wojtczak~\cite{TW10} showed that the reachability game and termination
(reachability with empty calling context) game problems are undecidable for RTAs with
two or more clocks. 
Moreover, they considered the so-called glitch-free restriction of RTAs---where at
each invocation either all clocks are passed by value or all clocks are passed
by reference--- and showed that the reachability (and termination) is
EXPTIME-complete for RTAs with two or more clocks. 
In the model of~\cite{TW10} it is compulsory to pass all the clocks at every
invocation with either mechanism. 
Abdulla, Atig, and Stenman~\cite{AAS12} studied a related model called
timed pushdown automata where they disallowed passing clocks by value. 
On the other hand, they allowed clocks to be passed either by reference or not
passed at all (in that case they are stored in the call context and continue to
tick with the uniform rate). 
It is shown in~\cite{AAS12} that the reachability problem for this class remains
decidable (EXPTIME-complete).
In this article, we restrict ourselves to the recursive timed automata model as
introduced in~\cite{TW10}.

Ouaknine and Worrell~\cite{OW10} proposed a thesis that restriction to
bounded-time recovers decidability for several key decision problem related to
real-time verification. 
In support of this thesis a number of important undecidable problems have been
shown to be decidable under bounded-time restriction, for instance language
inclusion for timed automata, emptiness problem for alternating timed automata,
and emptiness problem for rectangular hybrid automata.
The goal of this work was to approach reachability games on recursive timed
automata from this viewpoint and to recover the decidability of these games
under time-bounded restriction. 
However, we discovered a rather negative result. 
In this paper we show that the problem stays undecidable for RTA with just $3$
or more clocks.  

We also consider the extension of RTAs with stopwatches (clocks that can be
paused) to recursive stopwatch automata (RSAs) and show that the  time-bounded
reachability game problem stays undecidable even for RSAs with $3$
or more stopwatches,  while we show decidability of glitch-free RSAs with $2$
stopwatches.   
We also show that the reachability problem is undecidable for unrestricted RSA
with two or more stopwatches. 
For the time-bounded reachability case, we show that the problem stays
undecidable even for glitch-free variant of RSAs with $4$ or more stopwatches.   
The Table~\ref{tab:res} highlights our contributions related to reachability games
on recursive timed and stopwatch automata. 
For a survey of models related to RTA and dense-time
pushdown automata we refer the reader to~\cite{TW10} and~\cite{AAS12}.
\begin{table}[t]
  \begin{tabular}{lcccc}
    \hline
    &\multicolumn{2}{c}{~~~~~~~Recursive Timed Automata~~~~~~~~~~}&\multicolumn{2}{c}{
      ~~~~~~~~~Recursive Stopwatch Automata~~~~~~~~}\\
    &TUB& TB & TUB & TB\\
    \toprule
    Glitch-free&  {\color{gray} D} &  {\color{gray} D} &  \textbf{U} ($\geq 3$ sw)  &{\textbf U}($\geq 4$ sw)\\
    &       &     & \textbf{D} $(\leq 2$ sw)& \textbf{D} $(\leq 2$ sw)\\
    \hline
    Unrestricted &  {\color{gray} U ($\geq 2$ clocks)} &  \textbf{U} ($\geq 3$ clocks)&\textbf{U} $(\geq 2$ sw) &\textbf{U} ($\geq 3$ sw) \\
    \hline
  \end{tabular}
  \caption{Summary of the contributions of this paper. Results shown in bold are
    contributions from this paper, while results shown in gray color are
    from~\cite{TW10}.  
    Here TUB  and TB stand for time-unbounded and time-bounded reachability games, and $\text{U}$
    stands for undecidable, $\text{D}$ for decidable, and $\text{sw}$ for
    stopwatches.} 
\vspace{-1em}
\label{tab:res}
\end{table}
 Due to space limitations, we only sketch the key proofs and details can be found in \cite{KMT14}.

\section{Preliminaries}
\subsection{Reachability Games on Labelled Transition Systems.}
A \emph{labelled transition system} (LTS) is a tuple $\lts = (S, A, \trans)$
where $S$ is the set of \emph{states}, 
$A$ is the set of \emph{actions}, and $\trans : S {\times} A \to S$ is
the \emph{transition function}.
We say that an LTS $\lts$ is \emph{finite} (\emph{discrete}) if both $S$
and $A$ are finite (countable).
We write $A(s)$ for the set of actions available at $s \in S$, i.e., $A(s) =
\set{a \::\: \trans(s, a) \not = \emptyset}$. 
A \emph{game arena} $G$ is a tuple $(\lts, S_\Ach, S_\Tor)$, where $\lts = (S,
A, \trans)$ is an LTS, $S_\Ach \subseteq S$ is the set of states controlled by
player \ach, and $S_\Tor \subseteq S$ is the set of states controlled by $\tort$. 
Moreover, sets  $S_\Ach$ and $S_\Tor$ form a partition of the set
$S$.
In a \emph{reachability game} on $G$, rational players---\ach
and \tort---take turns to move a token along the states of $\lts$.
The decision to choose the successor state is made by the player
controlling the current state.
The objective of \ach is to eventually reach certain states, while the
objective of \tort is to avoid them forever.

We say that $(s, a, s') \in S {\times} A {\times} S$ is a transition of
$\lts$ if $s' = \trans(s, a)$ and a \emph{run} of $\lts$ is a sequence
$\seq{s_0, a_1, s_1, \ldots} \in S {{\times}} (A {{\times}} S)^*$
such that $(s_i, a_{i+1}, s_{i+1})$ is a transition of $\lts$ for all
$i \geq 0$.
We write $\RUNS^\lts$ ($\FRUNS^\lts$) for the sets of infinite (finite)
runs and $\RUNS^\lts(s)$ ($\FRUNS^\lts(s)$) for the
sets of infinite (finite) runs starting from state~$s$.
For a set $F \subseteq S$ and a run $r = \seq{s_0, a_1, \ldots}$ we
define $\STOP(F)(r) = \inf \set{ i\in \Nat \::\: s_i \in F}$.
Given a state $s \in S$ and a set of final states $F \subseteq S$ we
say that a final state is reachable from $s_0$ if there is a run $r \in
\RUNS^\lts(s_0)$ such that $\STOP(F)(r) < \infty$.
A strategy of \ach is a partial function
$\alpha :\FRUNS^\lts \to A$ such that for a run $r \in \FRUNS^\lts$ we
have that $\alpha(r)$ is defined if $\LAST(r) \in S_\Ach$, and
$\alpha(r) \in A(\LAST(r))$ for every such $r$.
A strategy of \tort is defined analogously.
Let $\LSigma_\Ach$ and $\LSigma_\Tor$ be the set of strategies of
\ach and \tort, respectively.
The unique run $\RUN(s, \alpha, \tau)$ from a state $s$ when players use
strategies $\alpha \in \LSigma_\Ach$ and $\tau \in \LSigma_\Tor$ is
defined in a straightforward manner.

Given an initial state $s$ and a set of final states $F$, 
and strategies $\tau$ for \tort and $\alpha$ for \ach, 
\ach is said to win the reachability game 
if $\STOP(F)(\RUN(s, \alpha, \tau)) < \infty$,
else if $\STOP(F)(\RUN(s, \alpha, \tau)) = \infty$, then \tort 
is the winner. 
A \emph{reachability game problem} is to decide whether in a given
game arena $G$, an initial state $s$ and a set of final states $F$,
\ach has a strategy to win the reachability game (irrespective of 
\tort's strategy).
\subsection{Reachability Games on Recursive state machines}
A \emph{recursive state machine}~\cite{ABEGRY05} $\Mm$ is a tuple 
$(\Mm_1, \Mm_2, \ldots, \Mm_k)$ of components, where each component $\Mm_i =
(N_i, \En_i, \Ex_i, B_i, Y_i, A_i, \trans_i)$ for each $1 \leq i \leq k$  is
such that:
\begin{itemize}
\item
  $N_i$ is a finite set of \emph{nodes} including a distinguished set $\En_i$ of
  \emph{entry nodes} and a set $\Ex_i$ of \emph{exit nodes} such that $\Ex_i$
  and $\En_i$ are disjoint sets; 
\item
  $B_i$ is a finite set of \emph{boxes};
\item
  $Y_i: B_i \to \set{1, 2, \ldots, k}$ is a mapping 
  that assigns every box to a component.
  We associate a set of \emph{call ports} $\call(b)$ and return ports
  $\return(b)$ to each box $b \in B_i$:
  \begin{eqnarray*}
  \call(b) = \set{(b, en) \::\: en \in \En_{Y_i(b)} } &\text{ and } &
  \return(b) = \set{(b, ex) \::\: ex \in  \Ex_{Y_i(b)}}.
  \end{eqnarray*}
  Let $\call_i = \cup_{b \in B_i} \call(b)$ and
    $\return_i = \cup_{b \in B_i} \return(b)$ be the set of call and
    return ports of component $\Mm_i$.
    We define the set of locations $Q_i$  of component $\Mm_i$ as the 
    union of the set of nodes, call ports and return ports, i.e. 
    $Q_i = N_i  \cup \call_i \cup \return_i$;
  \item
    $A_i$ is a finite set of \emph{actions}; and
  \item
    $\trans_i : Q_i {{\times}} A_i \to Q_i$ is the transition function  with a
    condition that call ports and exit nodes do not have any outgoing transitions.
  \end{itemize}
For the sake of simplicity, we assume that the set of boxes $B_1,
\ldots, B_k$ and set of nodes $N_1, N_2, \ldots, N_k$ are mutually
disjoint.
We use symbols $N, B, A, Q, \trans$, etc. to denote
the union of the corresponding symbols over all components.

\begin{figure}[t]
  \centering
 \vspace*{-0.4cm}
  \makebox[2cm]{{\small
  \begin{tikzpicture}[node distance=4cm]

    \draw(-2, -1) rectangle (2,1);
    \draw (1.8, 1.2) node {$M_1$};

    \node[loc](u1) at (-2, 0.5) {$u_1$};
    \node[loc](u2) at (-2, -0.5) {$u_2$};

    \node[loc](u4) at (2, 0) {$u_4$};

    \node[boxloc](b1) at (-0.4, 0.5) {$~b_1:M_2~$};
    \node[port](b1v1) at (-1.15, 0.6) {};
    \node[port](b1v2) at (-1.15, 0.4) {};
    \node[port](b1v3) at (0.31, 0.6) {};
    \node[port](b1v4) at (0.31, 0.4) {};

    \node[boxloc](b2) at (-0.4, -0.5) {$~b_2:M_3~$};
    \node[port](b2w1) at (0.35, -0.5) {};
    \node[port](b2w2) at (-1.15, -0.5) {};

    \node[loc](u3) at (1.15, -0.5) {$u_3$};

    \draw[trans] (u1)--(b1v1);
    \draw[trans] (u2)--(b2w2);
    \draw[trans] (u3)--(u4);
    \draw[trans] (b2w1)--(u3);
    \draw[trans] (b1v3)--(u4);
    \draw[trans] (b1v4)-- +(0.25, 0) -- +(0.25, -0.4) --+(-1.7, -0.4) --
    +(-1.7, 0) -- +(-1.5, 0);

    \draw(3, -1) rectangle (6,1);
    \draw (5.8, 1.2) node {$M_2$};

    \node[loc](v1) at (3, 0.5) {$v_1$};
    \node[loc](v2) at (3, -0.5) {$v_2$};
    \node[loc](v3) at (6, 0.5) {$v_3$};
    \node[loc](v4) at (6, -0.5) {$v_4$};

    \node[boxloc](c1) at (4.5, 0.5) {$~c_1:M_2~$};
    \node[port](c1v1) at (3.75, 0.6) {};
    \node[port](c1v2) at (3.75, 0.4) {};
    \node[port](c1v3) at (5.19, 0.6) {};
    \node[port](c1v4) at (5.19, 0.4) {};

    \node[boxloc](c2) at (4.5, -0.5) {$~c_2:M_3~$};
    \node[port](c2w1) at (3.8, -0.5) {};
    \node[port](c2w2) at (5.19, -0.5) {};

    \draw[trans] (v1)--(c1v1);
    \draw[trans] (v2)-- + (0.2, 0.6) -- (c1v2);
    \draw[trans] (v2) -- (c2w1);
    \draw[trans] (c1v3) -- +(0.2, 0)--(v4);
    \draw[trans] (c1v4)--(v3);
    \draw[trans] (c2w2) --(v4);
    \draw[trans] (c2w2)-- +(0.2, 0.2) -- +(0.2, 0.5) -- +(-1.6, 0.5) --(c1v2);


    \draw(7, -1) rectangle (9.5,1);
    \draw (9.3, 1.2) node {$M_3$};

    \node[loc](w1) at (7, 0) {$w_1$};
    \node[loc](w2) at (9.5, 0) {$w_2$};

    \node[boxloc](d) at (8.3, 0.5) {$~d:M_1~$};
    \node[port](du1) at (7.65, 0.6) {};
    \node[port](du2) at (7.65, 0.4) {};
    \node[port](du4) at (8.95, 0.5) {};

    \draw[trans] (w1) .. controls +(60:0.5) and +(180:0.8) .. (du1);
    \draw[trans] (du4) -- (w2);
    \draw[trans] (w1) edge[bend right=25] (w2);

  \end{tikzpicture}
}}
  \vspace*{-0.4cm}
  \caption{Example recursive state machine taken from~\cite{ABEGRY05}}
  \vspace*{-0.4cm}
 \label{fig:example}
\end{figure}
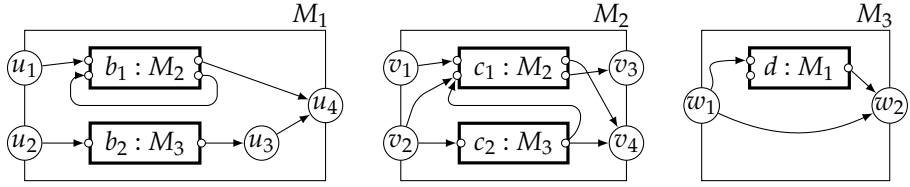
An example of a RSM is shown in Figure~\ref{fig:example} (taken from \cite{TW10}).
An execution of a RSM begins at the entry node of some component and depending
upon the sequence of input actions the state evolves naturally like a labelled
transition system. 
However, when the execution reaches an entry port of a box, this box is stored
 on a stack of pending calls, and the execution continues naturally from the
 corresponding entry node of the component mapped to that box. 
 When an exit node of a component is encountered, and if the stack of pending
 calls is empty then the run terminates; otherwise, it pops the box from the top
 of the stack and jumps to the exit port of the just popped box corresponding
 to the just reached exit of the component. 
We formalize the semantics of a RSM using a discrete LTS, whose states are pairs
consisting of a sequence of boxes, called the context, mimicking the stack of
pending calls and the current location.
Let $\Mm = (\Mm_1, \Mm_2, \ldots, \Mm_k)$ be an RSM where the component
$\Mm_i$ is $(N_i, En_i, Ex_i, B_i, Y_i, A_i, \trans_i)$.
The semantics of $\Mm$ is the discrete  labelled transition system
$\sem{\Mm} = (S_\Mm, A_\Mm, \trans_\Mm)$ where:
\begin{itemize}
\item
  $S_\Mm \subseteq B^* {\times} Q$ is the set of states;
\item
  $A_\Mm = \cup_{i=1}^{k} A_i$ is the set of actions;
\item
  $\trans_\Mm : S_\Mm {\times} A_\Mm \to S_\Mm$ is the transition
  function such that for $s = (\sseq{\kappa}, q) \in S_\Mm$ and
  $a \in A_\Mm$, we have that $s' =  \trans_\Mm(s, a)$ if and only
  if one of the following holds:
  \begin{enumerate}
  \item
    the location $q$ is a call port, i.e. $q = (b, en) \in \call$,
    and $s' = (\sseq{\kappa, b}, en)$;
  \item
    the location $q$ is an exit node, i.e. $q = ex \in \Ex$
    and $s' = (\sseq{\kappa'}, (b, ex))$ where
    $(b, ex) \in \return(b)$ and $\kappa = (\kappa', b)$;
  \item
    the location $q$ is any other kind of location, and
    $s' = (\sseq{\kappa}, q')$ and $q' \in \trans(q, a)$.
  \end{enumerate}
\end{itemize}

Given $\Mm$ and a subset $Q' \subseteq Q$ of its nodes we define 
$\sem{Q'}_\Mm$ as  $\set{(\sseq{\kappa}, v') \::\: \kappa \in B^* \text{ and }
  v' \in Q'}$.
We define the terminal configurations $\term_\Mm$ as
the set $\set{(\sseq{\varepsilon}, ex) \::\: ex \in \Ex}$ with the empty context 
$\sseq{\varepsilon}$.
Given a recursive state machine  $\Mm$, an initial node $v$, and a set
of \emph{final locations} $F \subseteq Q$ the {\em reachability problem}
on $\Mm$ is defined as the reachability problem on the LTS $\sem{\Mm}$ with
the initial state $(\sseq{\varepsilon}, v)$ and final states $\sem{F}$.
We define \emph{termination problem} as the reachability of one of the exits with the empty context.
The reachability and the termination problem for recursive state machines can
be solved in polynomial time~\cite{ABEGRY05}. 

A partition $(Q_\Ach, Q_\Tor)$ of locations $Q$ of an RSM $\Mm$
(between \ach and \tort) gives rise to recursive game arena
$G =(\Mm, Q_\Ach, Q_\Tor)$.
Given an initial state, $v$, and a set of final states, $F$, the
reachability game on $\Mm$ is defined as the reachability game on the game
arena $(\sem{\Mm}, \sem{Q_\Ach}_\Mm, \sem{Q_\Tor}_\Mm)$ with the initial state
$(\sseq{\varepsilon}, v)$ and the set of final states $\sem{F}_\Mm$.
Also, the termination game $\Mm$ is defined as the reachability game on the
game arena $(\sem{\Mm}, \sem{Q_\Ach}_\Mm, \sem{Q_\Tor}_\Mm)$ with the initial
state $(\sseq{\varepsilon}, v)$ and the set of final states $\term_\Mm$.
It is a well known result (see, e.g.~\cite{Wal96,Ete04}) that
reachability games and termination games on RSMs are decidable
(EXPTIME-complete).

\section{Recursive Hybrid Automata}
\label{sec:recursive-timed-automata}
In this paper since we study both recursive timed automata as well as recursive
stopwatch automata, we introduce a more general recurisve hybrid automata, and
from that we define these two subclasses.
Recursive hybrid automata (RHAs)  extend classical hybrid automata (HAs) with
recursion in a similar way RSMs extend LTSs.
We introduce a rather simpler subclass of hybrid automata known as singular
hybrid automata where all variables grow with constant-rates. 
\subsection{Syntax}
 Let $\Real$ be the set of real numbers.
 Let $\variables$ be a finite set of real-valued variables.
 A \emph{valuation} on $\variables$ is a function $\nu : \variables \to \Real$.
 We assume an arbitrary but fixed ordering on the variables and write $x_i$
 for the variable with order $i$. 
 This allows us to treat a valuation $\nu$ as a point $(\nu(x_1), \nu(x_2),
 \ldots, \nu(x_n)) \in \Real^{|\variables|}$. 
 Abusing notations slightly, we use a valuation on $\variables$ and a point in
 $\Real^{|\variables|}$ interchangeably. 
For a subset of variables $X \subseteq \variables$ and a valuation
$\nu' \in \variables$, we write $\nu[X{:=}\nu']$ for the valuation where
$\nu[X{:=}\nu'](x) = \nu'(x)$ if $x \in X$, and
$\nu[X{:=}\nu'](x) = \nu(x)$ otherwise.
The valuation $\zero \in \V$ is a special valuation such that
$\zero(x) = 0$ for all $x \in \variables$.

 We define a constraint over a set $\variables$ as a subset of $\Real^{|\variables|}$.
 We say that a constraint is \emph{rectangular} if it is defined as the conjunction
 of a finite set of constraints of the form 
 $x  \bowtie k,$  where $k \in \Int$, $x \in \variables$, and $\bowtie \in \{<,\leq, =,
 >, \geq\}$.    
 For a constraint $G$, we write $\sem{G}$ for the set of valuations in
 $\Real^{|\variables|}$ satisfying the constraint $G$.  
 We write $\top$ ( resp., $\bot$) for the special constraint that is true
 (resp., false) in all the valuations, i.e. $\sem{\top} = \Real^{|\variables|}$ 
(resp., $\sem{\bot} = \emptyset$). 
 We write $\rect(\variables)$ for the set of rectangular constraints over $\variables$ including
 $\top$ and $\bot$.  
\begin{definition}[Recursive Hybrid Automata]
  A \emph{recursive hybrid automaton} $\Hh = (\variables, (\Hh_1, \Hh_2, \ldots,
  \Hh_k))$ is a pair made of a set of variables $\variables$ and a collection of
  components $(\Hh_1, \Hh_2, \ldots, \Hh_k)$ where every component 
  $\Hh_i =  (N_i, \En_i, \Ex_i, B_i, Y_i, A_i, \trans_i, P_i, \inv_i,
  E_i, J_i, F_i)$ is such that:
  \begin{itemize}
  \item
    $N_i$ is a finite set of \emph{nodes} including a distinguished set $\En_i$ of
    \emph{entry nodes} and a set $\Ex_i$ of \emph{exit nodes} such that $\Ex_i$
    and $\En_i$ are disjoint sets; 
  \item
    $B_i$ is a finite set of \emph{boxes};
  \item
    $Y_i: B_i \to \set{1, 2, \ldots, k}$ is a mapping that assigns every box to
    a component. 
    (Call ports $\call(b)$ and return ports $\return(b)$ of a box $b \in B_i$,
    and call ports $\call_i$ and return ports $\return_i$ of a component $\Hh_i$
    are defined as before. 
    We set $Q_i = N_i \cup \call_i \cup \return_i$ and refer to this set as the
    set of locations of $\Hh_i$.) 
  \item
    $A_i$ is a finite set of \emph{actions}.
  \item
    $\trans_i : Q_i {{{\times}}} A_i \to Q_i$ is the transition function  with a
    condition that call ports and exit nodes do not have any outgoing transitions.
  \item
    $P_i: B_i \to 2^\variables$ is pass-by-value mapping that assigns every box
    the set of variables that are passed by value to the component
    mapped to the box; (The rest of the variables are assumed to be passed by
    reference.)
  \item
    $\inv_i : Q_i \to \rect(\variables)$ is the \emph{invariant condition};
  \item
    $E_i : Q_i {{\times}} A_i \to \rect(\variables)$ is the \emph{action enabledness function};  
  \item
    $J_i : A_i \to 2^{\variables}$ is the \emph{variable reset function}; and
  \item
    $F_i : Q_i \to \Nat^{|\variables|}$ is the \emph{flow function}
    characterizing the rate of each variable in each location.
  \end{itemize}
  We assume that the sets of boxes, nodes, locations, etc. are mutually
  disjoint across components and we write ($N, B, Y, Q, P, \trans$, etc.) to denote corresponding
  union over all components. 
\end{definition}
We say that a recursive hybrid automaton is \emph{glitch-free} if for
every box either all variables are passed by value or none is passed by
value, i.e. for each $b \in B$ we have that either $P(b) = \variables$
or $P(b) = \emptyset$. 
Any general recursive hybrid automaton with one variable is trivially
glitch-free.
We say that a RHA is \emph{hierarchical} if there exists an ordering over
components such that a component never invokes another component of higher
order or same order. 

We say that a variable $x \in \variables$ is a \emph{clock} (resp., a
stopwatch) if for every location $q \in Q$ we have that $F(q)(x) = 1$ (resp.,
$F(q)(x) \in \set{0, 1}$). 
A recursive timed automaton (RTA) is simply a recursive hybrid automata where all
variables $x \in \variables$ are clocks. 
Similarly, we define a recursive stopwatch automaton (RSA) as a recursive hybrid
automaton where all variables $x \in \variables$ are stopwatches. 
Since all of our results pertaining to recursive hybrid automata are shown in
the context of recursive stopwatch automata, we often confuse RHA with RSA. 

\subsection{Semantics}
A \emph{configuration} of an RHA $\Hh$ is a tuple $(\sseq{\kappa},
q, \val)$, where $\kappa \in (B {\times} \V)^*$ is 
sequence of pairs of boxes and variable valuations, $q \in Q$ is a
location and $\val \in \V$ is a variable valuation over $\variables$ such that
$\val \in \inv(q)$.
The sequence $\sseq{\kappa} \in (B {\times} \V)^*$ denotes the stack of
pending recursive calls and the valuation of all the variables at the
moment that call was made, and we refer to this sequence as the context of
the configuration.
Technically, it suffices to store the valuation of variables passed by
value, because other variables retain their value after returning from a
call to a box, but storing all of them simplifies the notation.
We denote the empty context by $\sseq{\epsilon}$.
For any $t \in \Real$, we let $(\sseq{\kappa}, q, \val){+}t$
equal the configuration $(\sseq{\kappa}, q, \val{+} F(q) \cdot t)$.
Informally, the behaviour of an RHA is as follows.
In configuration $(\sseq{\kappa}, q, \val)$ time passes before an available
action is triggered, after which a discrete transition occurs.
Time passage is available only if the invariant condition $\inv(q)$
is satisfied while time elapses, and an action $a$ can be chosen after
time $t$ elapses only if it is enabled after time elapse, i.e.,\
if $\val{+} F(q) \cdot t \in E(q, a)$.
If the action $a$ is chosen then the successor state is
$(\sseq{\kappa}, q', \val')$ where $q' \in X(q,a)$ and
$\val' = (\val+t)[J(a) := \mathbf{0}]$.
Formally, the semantics of an RHA is given by an LTS which has both an
uncountably infinite number of states and transitions.

\begin{definition}[RHA semantics]
  Let $\Hh = (\variables, (\Hh_1, \Hh_2, \ldots, \Hh_k))$ be an RHA where
  each component is of the form
  $\Hh_i =  (N_i, \En_i, \Ex_i, B_i, Y_i, A_i, \trans_i, P_i, \inv_i,
  E_i, J_i, F_i)$.
  The semantics of $\Hh$ is a labelled transition system
  $\sem{\Hh} = (\nS, \nA, \ntrans)$ where:
  \begin{itemize}
  \item
    $\nS \subseteq (B{\times} \V)^* {\times} Q\, {\times}\, \V$, the set of
    states, is s.t. 
    $(\sseq{\kappa}, q, \val) {\in} \nS$ if $\val {\in} \inv(q)$.
  \item
    $\nA = \Rplus {{\times}} A$ is the set of \emph{timed actions}, where $\Rplus$ is the set of non-negative reals;
  \item
    $\ntrans : \nS {\times} \nA \to \nS$ is the transition function
    such that for $(\sseq{\kappa}, q, \val) \in \nS$ and
    $(t, a) \in \nA$, we have
    $(\sseq{\kappa'}, q', \val') = \ntrans((\sseq{\kappa}, q, \val),
    (t, a))$ if and only if the following condition holds:
    \begin{enumerate}
    \item
      if the location $q$ is a call port, i.e. $q = (b, en) \in
      \call$ then $t = 0$, the context
      $\sseq{\kappa'} = \sseq{\kappa, (b, \val)}$,
      $q' = en$, and $\val' = \val$.
    \item
      if the location $q$ is an exit node, i.e. $q = ex \in Ex$,
      $\sseq{\kappa} = \sseq{\kappa'', (b, \val'')}$,
      and let $(b, ex) \in \return(b)$, then
      $t = 0$; $\sseq{\kappa'} = \sseq{\kappa''}$;
      $q' {=} (b, ex)$; and $\val' {=} \val[P(b){:=}\val'']$.
    \item
      if location $q$ is any other kind of location, then  $\sseq{\kappa'} =
      \sseq{\kappa}$,  $q' \in \trans(q,a)$, and
 \begin{enumerate} 
      \item $\val {+} F(q)\cdot t' \in \inv(q)$ for all $t' \in [0, t]$;
      \item $\val {+} F(q)\cdot t \in E(q, a)$; 
      \item $\val' = (\val + F(q){\cdot} t)[J(a):= \zero]$.
      \end{enumerate}
    \end{enumerate}
  \end{itemize}
\end{definition}

\subsection{Reachability and Time-Bounded Reachability Game Problems}
For a subset $Q' \subseteq Q$ of states of RHA $\Hh$ we
define the set $\sem{Q'}_\Hh$ as the set $\set{(\sseq{\kappa}, q, \val) \in S_\Hh
  \::\: q \in Q'}$. 
We define the terminal configurations as
${\term_\Hh = \set{(\sseq{\varepsilon}, q, \val) \in S_\Hh \::\:   q \in \Ex}}$.
Given a recursive hybrid automaton $\Hh$, an initial node $q$ and
valuation $\val \in \V$, and a set of \emph{final locations} $F \subseteq
Q$, the {\em reachability problem} on $\Hh$ is to decide the existence of a run
in the LTS $\sem{\Hh}$ staring from the initial state $(\sseq{\varepsilon}, q,
\val)$ to some state in $\sem{F}_\Hh$.
As with RSMs, we also define \emph{termination problem} as reachability of one
of the exits with the empty context. 
Hence, given an RHA $\Hh$ and an initial node $q$ and a valuation $\val
\in \V$, the termination problem on $\Hh$ is to decide the existence of a run in the
LTS $\sem{\Hh}$ from initial state $(\sseq{\varepsilon}, q, \val)$ to a final
state in  $\term_\Hh$.

Given a run $r = \seq{s_0, (t_1, a_1), s_2, (t_2, a_2), \ldots, (s_n, t_n)}$ of an
RHA, its time duration $\Time(r)$ is defined as $\sum_{i=1}^{n}t_i$. 
Given a recursive hybrid automaton $\Hh$, an initial node $q$, a bound $T \in
\Nat$, and valuation $\val \in \V$, and a set of \emph{final locations} $F \subseteq
Q$, the {\em time-bounded reachability problem} on $\Hh$ is to decide the
existence of a run $r$ in the LTS $\sem{\Hh}$ staring from the initial state
$(\sseq{\varepsilon}, q, \val)$ to some state in $\sem{F}_\Hh$ such that
$\Time(r) \leq T$. 
Time-bounded termination problem is defined in an analogous manner. 

A partition $(Q_\Ach, Q_\Tor)$ of locations $Q$ of an RHA $\Hh$
gives rise to a recursive hybrid game arena 
$\Gamma = (\Hh, Q_\Ach, Q_\Tor)$.
Given an initial location $q$, a valuation $\nu \in V$ and a set of final
states $F$, the reachability game on $\Gamma$ is defined as the
reachability game on the game arena $(\sem{\Hh}, \sem{Q_\Ach}_\Hh,
\sem{Q_\Tor}_\Hh)$ with the initial state
$(\sseq{\varepsilon}, (q, \nu))$ and the set of final states $\sem{F}_\Hh$.
Also, termination game on $\Gamma$ is defined as the reachability game on
the game arena $(\sem{\Hh}, \sem{Q_\Ach}_\Hh, \sem{Q_\Tor}_\Hh)$ with
the initial state $(\sseq{\varepsilon}, (q, \nu))$ and the set of final states
$\term_\Hh$.

We prove the following key theorem about reachability games on various
subclasses of recursive hybrid automata in Section \ref{undec}.
\begin{theorem}
\label{th:main}
The reachability game problem is undecidable for: 
\begin{enumerate}
\item 
 Unrestricted RSA with 2 stopwatches,
\item 
 Glitch-free RSA with 3 stopwatches, 
\item 
 Unrestricted RTA with 3 clocks under bounded time, and 
\item 
 Glitch-free RSA with 4 stopwatches under bounded time.
\end{enumerate}
Moreover, all of these results hold even under hierarchical restriction. 
\end{theorem}

On a positive side, we observe that for glitch-free RSA with two stopwatches 
reachability games are decidable by exploiting the existence of finite bisimulation
for hybrid automata with 2 stopwatches. Details can be found in \cite{KMT14}.
\begin{theorem}
 \label{thm:dec}
 The reachability games are decidable for glitch-free RSA
 with atmost two stopwatches. 
\end{theorem}

\section{Undecidability Results}
\label{undec}
 
In this section, we provide a proof sketch of our undecidability results by
reducing the halting problem for two counter machines  to the reachability
problem in an RHA/RTA. 
Before we show our reduction, we give a formal definition of two-counter
machines. 

\begin{definition}[Two-counter Machines]
  A two-counter machine is a tuple $(L, C)$ where ${L = \set{\ell_0,
      \ell_1, \ldots, \ell_n}}$ is the set of instructions---including a
  distinguished terminal instruction $\ell_n$ called HALT---and ${C =
    \set{c_1, c_2}}$ is the set of two \emph{counters}.  
  The instructions $L$ are of the type:
  \begin{enumerate}
  \item (increment $c$) $\ell_i : c := c+1$;  goto  $\ell_k$,
  \item (decrement $c$) $\ell_i : c := c-1$;  goto  $\ell_k$,
  \item (zero-check $c$) $\ell_i$ : if $(c >0)$ then goto $\ell_k$
    else goto $\ell_m$,
  \item (Halt) $\ell_n:$ HALT.
  \end{enumerate}
  where $c \in C$, $\ell_i, \ell_k, \ell_m \in L$.
\end{definition}
A configuration of a two-counter machine is a tuple $(l, c, d)$ where
$l \in L$ is an instruction, and $c, d$ are natural numbers that specify the value
of counters $c_1$ and $c_2$, respectively.
The initial configuration is $(\ell_0, 0, 0)$.
A run of a two-counter machine is a (finite or infinite) sequence of
configurations $\seq{k_0, k_1, \ldots}$ where $k_0$ is the initial
configuration, and the relation between subsequent configurations is
governed by transitions between respective instructions.
The run is a finite sequence if and only if the last configuration is
the terminal instruction $\ell_n$.
Note that a two-counter  machine has exactly one run starting from the initial
configuration. 
The \emph{halting problem} for a two-counter machine asks whether 
its unique run ends at the terminal instruction $\ell_n$.
The halting problem~\cite{Min67} for two-counter machines is undecidable.

For all the undecidability results, we construct a recursive automaton
(timed/hybrid) as per the case, whose main components are the modules for the
instructions and  the counters are encoded in the variables of the automaton.  
In these reductions,  the reachability of the exit node of each component
corresponding to an instruction is linked to a faithful simulation of various
increment, decrement and zero check  instructions of the machine by choosing
appropriate delays to adjust the clocks/variables, to reflect changes  in
counter values. 
We specify a main component for each type instruction of the two counter
machine, for example $\Hh_{inc}$ for increment.  
The entry node and exit node of a main component $\Hh_{inc}$ corresponding to an 
instruction [$\ell_i:  c := c +1$;  goto  $\ell_k$] 
are respectively $\ell_i$ and $\ell_k$.  
 Similarly, a main component corresponding to a zero check instruction 
[$l_i$: if $(c >0)$ then goto $\ell_k$]
  else goto $\ell_m$, has a unique entry node  $\ell_i$, and two exit nodes corresponding to $\ell_k$ and 
  $\ell_m$ respectively. 
The various main components corresponding to the various instructions, when
connected appropriately, gives the higher level component $\Hh_{M}$ and this
completes the RHA $\Hh$.  
The entry node of $\Hh_{M}$ is the entry node of the main component for the
first instruction of $M$ and the exit node is $HALT$. 
\ach simulates the machine while \tort verifies the simulation. Suppose in each
main component for each type of instruction correctly \ach simulates the
instruction by accurately updating the counters encoded in the variables of
$\Hh$. Then, the unique run in $M$ corresponds to an unique run in
$\Hh_{M}$. The halting  problem of the two counter machine now boils down to
existence of a \ach strategy to ensure the reachability of an exit node $HALT$
(and $\ddot \smile$) in $\Hh_{M}$.\footnote{Hence the set of final nodes is node $HALT$ and all nodes labelled $\ddot \smile$.}

For the correctness proofs, we represent runs in the RSA
using three different forms of transitions
$s \stackrel[t]{g,J}{\longrightarrow} s'$, $s  \rightsquigarrow s'$ and $s \stackrel[M(V)]{*}{\longrightarrow} s'$ 
defined in the following way:
\begin{enumerate}
\item The transitions of the form $s \stackrel[t]{g,J}{\longrightarrow} s'$, where $s=(\langle \kappa \rangle, n, \nu)$, 
  $s'=(\langle \kappa \rangle, n', \nu')$ are configurations of the RHA, $g$ is a constraint or guard 
  on variables that enables the transition,
  $J$ is a set of variables, and $t$ is a real number, holds if there is a transition in the RHA from vertex $n$ to $n'$ with guard $g$ and reset 
  set  $J$. Also, $\nu'=\nu+rt[J:=0]$, where $r$ is the rate vector 
  of state $s$.
\item The transitions of the form $s \rightsquigarrow s'$ where $s=(\langle \kappa \rangle,n,\nu)$,
  $s'=(\langle \kappa' \rangle, n', \nu')$ correspond to the following cases:
  \begin{itemize}
  \item transitions from a call port to an entry node. That is, $n=(b,en)$ for some box
  $b \in B$ and $\kappa' =\langle \kappa, (b, \nu)\rangle$ and $n'=en \in \En$ while $\nu'=\nu$. 
  \item transitions from an exit node to a return port which restores values of the variables passed by value, that is, $\langle \kappa \rangle=\langle \kappa'',(b, \nu'')\rangle$, $n=ex \in \Ex$ and $n'=(b,ex) \in Ret(b)$ and $\kappa'=\kappa''$, while $\nu'=\nu[P(b):=\nu'']$.
  \end{itemize}
  \item The transitions of the form $s \stackrel[M(V)]{t}{\longrightarrow} s'$, called summary edges, where 
  $s=(\langle \kappa \rangle, n, \nu)$, $s'=(\langle \kappa \rangle, n', \nu')$ are such that $n=(b,en)$ and $n'=(b,ex)$ are call and return ports, respectively, of a box $b$ mapped
to $M$ which passes by value to $M$, the variables in $V$. $t$ is the time elapsed between the occurences of $(b,en)$ and $(b,ex)$. In other words, $t$ is the time elapsed in the component $M$.
    \end{enumerate}
  A configuration $(\langle \kappa \rangle, n, \nu)$ is also written as $(\langle \kappa \rangle, n, (\nu(x),\nu(y)))$.  
  
\subsection{Time Bounded Reachability Games in Unrestricted RTA}
\label{undec:rtg_3clk}

\begin{lemma}
\label{3clk}
The time bounded reachability game problem is undecidable for recursive timed automata 
  with at least 3 clocks.
\end{lemma}
\begin{proof}
We prove that the reachability problem is undecidable for  unrestricted RTA with 3 clocks.  
In order to obtain the undecidability result, we 
use a reduction from the halting problem 
for two counter machines. 
Our reduction uses a RTA with three clocks $x,y,z$.

We specify a main component for each instruction of the two counter machine. On entry into 
a main component for increment/decrement/zero check, 
we have $x=\frac{1}{2^{k+c}3^{k+d}}$, $y=\frac{1}{2^k}$ and $z=0$,
where $c,d$ are the current values of the counters and $k$ is the current instruction. Note that $z$ is used only to enforce urgency in several locations. 
Given a two counter machine, we build a 3 clock RTA 
whose building blocks are the main components for the instructions.   
The purpose of the components is to simulate faithfully the counter machine 
by choosing appropriate delays  to adjust the variables to reflect changes 
in counter values. 
On entering the entry node $en$ of a main component corresponding to an instruction $l_i$,
we have the configuration $(\langle \epsilon \rangle, en, (\frac{1}{2^{k+c}3^{k+d}},\frac{1}{2^{k}},0))$ 
of the three clock RTA.  

We discuss the module for incrementing counter $c$ here; more details can be found in \cite{KMT14}.
In all the components, the variables passed by value are written below the boxes and the invariants of the 
locations are indicated below them.

\noindent{\bf{Simulate increment instruction}}: Lets consider the increment instruction  
$\ell_i$: $c=c+1$; goto $\ell_k$. The component for this instruction is component $Inc~c$ given in Figure 
\ref{fig_3clk_rtg_inc}. Assume that $x=\frac{1}{2^{k+c}3^{k+d}}$, $y=\frac{1}{2^k}$ and $z=0$ at 
the entry node $en_1$ of the component $Inc~c$. To correctly simulate the increment of counter $c$, the clock values at the exit node $ex_1$ should be $x=\frac{1}{2^{k+c+2}3^{k+d+1}}$, $y=\frac{1}{2^{k+1}}$ and $z=0$. The value of $x$ goes from $x=\frac{1}{2^{k+c}3^{k+d}}$ to $x=\frac{1}{2^{k+c+2}3^{k+d+1}}$ so that $c$ is incremented and end of current $k+1$ instruction is also recorded. Thus $x=\frac{1}{2^{k+1+c+1}3^{k+1+d}}$.

Let $\alpha=\frac{1}{2^{k+c}3^{k+d}}$ and $\beta=\frac{1}{2^k}$.  We want  $x = \frac{\alpha}{12}$ and $y = \frac{\beta}{2}$ at $ex_1$. We utilise the component $Div\{a,n\}$ (instantiating $Div\{a,n\}$ with $a=x,n=12$ to achieve $x = \frac{\alpha}{12}$ and with $a=y,n=2$ to achieve $y=\frac{\beta}{2}$) to perform these divisions. 
Lets walk through the working of the component $Inc~c$. As seen above, at the entry node $en_1$, we have 
$x=\frac{1}{2^{k+c}3^{k+d}}$, $y=\frac{1}{2^k}$ and $z=0$. 

\begin{enumerate}
\item No time is spent at $en_1$ due to the invariant $z=0$. 
$Div\{y,2\}$ is called, passing $x,z$ by value.  
At the call port  of $A_1:Div\{y,2\}$, we have the same values of $x,y,z$. Let us examine the component $Div\{y,2\}$. 
We instantiate $Div\{a,n\}$ with $a=y,n=2$. Thus, the clock referred to as $b$ 
in $Div\{a,n\}$ is $x$ after the instantiation. 
At the entry node $en_2$ of $Div\{y,2\}$, no time is spent due to the invariant $z=0$; we have 
$a=y=\beta,b=x=\alpha,z=0$.  
Resetting $b$(i.e; $x$), we are at the call port of $A_3:D$. 
$A_3$ is called, passing $a,z$ by value. 
A nondeterministic time $t$ is spent at the entry node $en_3$ of $D$. 
 Thus, at the return port of $A_3$, we have $a=y=\beta,b=x=t,z=0$. 
The return port of $A_3$ is a node belonging to \tort; for \ach to 
reach $\ddot \smile$, $t$ must be $\frac{\beta}{2}$. \tort has 
two choices to make at the return port of $A_3$: he can continue the simulation, 
by resetting $a$(i.e; $y$) and going to the call port of $A_5:D$, 
or he can verify if $t$ is indeed $\frac{\beta}{2}$, by going to the call port of $A_4$. 

\begin{itemize}
\item Assume \tort goes to the call port of $A_4: C^{x=}_{y/2}$ (recall, that by the instantiation, $b=x$, $a=y$ and 
$n=2$). $z$ is passed by value. At the entry node $en_5$ of  
$C^{x=}_{y/2}$, no time elapses due to the invariant $z=0$. 
Thus, we have $x=b=t$, $a=y=\beta,z=0$ at $en_5$. 
The component $A_7:M_x$ is 
invoked, passing $x,z$ by value. At the entry node $en_6$ of $M_b$, a time $1-t$ 
is spent, giving $a=y=\beta+1-t$, $b=x=t$ and $z=0$ at the return port of $A_7$. 
Since $n=2$, one more invocation of $A_7: M_x$ is made, obtaining 
$a=y=\beta+2(1-t)$, $b=x=t$ and $z=0$ at the return port 
of $A_7$ after the second invocation. To reach the 
exit node $ex_5$
of $C^{x=}_{y/2}$, $a$ must be exactly 2, since 
no time can be spent at the return port of $A_7$; this is so since the invariant $z=0$ at the exit node $ex_5$ of 
$C^{x=}_{y/2}$ is satisfied only when no time 
is spent at the return port of $A_7$.  
If $a$ is exactly 2, we have 
 $\beta=2t$. In this case, 
from the return port of $A_4$,  
$\ddot \smile$ can be reached.  
 
 \item Now consider the case that \tort moves ahead from 
   the return port of $A_3$, resetting $a$(i.e; $y$) to the call port of $A_5:D$.
 The values are $a=y=0,b=x=t=\frac{\beta}{2}$ and $z=0$. 
   $A_5:D$ is invoked passing $b=x$ and $z$ by value. 
   A non-deterministic amount of time $t'$ is spent at the entry node $en_3$ of 
   $D$, giving $a=y=t'$, $b=x=\frac{\beta}{2}$ and $z=0$ 
   at the return port of $A_5$.  Again, the return port of $A_5$ is a node 
   belonging to \tort. Here \tort, thus has two choices: he can continue 
   with the simulation going to $ex_2$, or can 
   verify that $t'=\frac{\beta}{2}$ by going to the call port of $A_6:C^{y=}_x$.
   $C^{y=}_x$ is a component that checks if $y$ has ``caught up'' with $x$; that is, 
   whether $t'=t=\frac{\beta}{2}$. At the entry node $en_4$ of $C^{y=}_x$,
   $a$ and $b$ can simultaneeously reach 1 iff $t=t'$; that is, $t'=\frac{\beta}{2}$.
   Then, from the return port of $A_6$, we can reach $\ddot \smile$.   
   
  \item Thus, we reach $ex_2$ with $x=y=\frac{\beta}{2},z=0$. At the return port of 
  $A_1:Div\{y,2\}$, we thus have $x=\alpha, y=\frac{\beta}{2},z=0$.       
\end{itemize}
\item From the return port of $A_1:Div\{y,2\}$, we reach the call port of $A_2:Div\{x,12\}$.
$y,z$ are passed by value. The functioning of $A_2$ is similar to that of $A_1$:
at the return port of $A_1$, we obtain $x=\frac{\alpha}{12}$, $y=\frac{\beta}{2}$ and $z=0$. 
\end{enumerate}

\begin{figure}
\begin{center}
\begin{tikzpicture}[node distance=4cm] 
\tikzstyle{every node}=[font=\scriptsize]
\draw(0, 0) rectangle (7,1.5);
\draw (6.5, 1.7) node {$\incc$};

  \node[loc](n1) at (0, 0.7) {$en_{1}$};  
  \node[inv]() [below of =n1, node distance=5mm]{$\figinv{z{=}0}$};
  
  \node[boxloc](a1) at (2, 0.7) {$\begin{array}{c} A_1{:}\Div{y}{2}\end{array}$};
  \node[varpass] () [below of =a1, node distance=5mm]{$(x,z)$};
  \node[port](a1n1) [left of=a1,node distance = 10mm] {}; 
  \node[port](a1x1) [right of=a1,node distance = 10mm] {};

  \node[boxloc](a2) at (5, 0.7) {$\begin{array}{c} A_2{:}\Div{x}{12}\end{array}$};
  \node[varpass] () [below of =a2, node distance=5mm]{$(y,z)$};
  \node[port](a2n1) [left of=a2,node distance = 11mm] {}; 
  \node[port](a2x1) [right of=a2,node distance = 11mm] {};
  
  \node[loc](x1) at (7, 0.7) {$ex_{1}$};
  \node[inv]() [below of =x1, node distance=5mm]{$\figinv{z{=}0}$};
  
  \draw[trans] (n1)--(a1n1);
  \draw[trans] (a1x1) -- (a2n1);
  \draw[trans] (a2x1) -- (x1);

  \draw(0, -3.5) rectangle (6,-0.5);
  \draw (4.5, -0.3) node {$\Div{a}{n}:a,b\in\set{x,y}$};

  \node[loc](n2) at (0, -1.2) {$en_{2}$};  
  \node[inv]() [below of =n2,node distance=5mm] {$\figinv{z{=}0}$};  
  
  \node[boxloc](a3) at (1.5, -1.2) {$\begin{array}{c} A_3{:}\Del\end{array}$};
  \node[varpass] () [below of =a3, node distance=5mm]{$(a,z)$};
  \node[port](a3n1) [left of=a3,node distance = 5.5mm] {}; 
   \node[oport](a3x1) [right of=a3,node distance =5.5mm] {};  
  
  \node[boxloc](a4) at (1.5, -2.3) {$\begin{array}{c} A_4{:}\chkDiv{b}{a}{n}\end{array}$};
  \node[varpass] () [below of =a4, node distance=5mm]{$(z)$};
  \node[port](a4n1) [left of=a4,node distance = 7.2mm] {}; 
  \node[port](a4x1) [right of=a4,node distance = 7.2mm] {}; 
  
  \node[boxloc](a5) at (4.5, -1.2) {$\begin{array}{c} A_5{:}\Del\end{array}$};
  \node[varpass] () [below of =a5, node distance=5mm]{$(b,z)$};
  \node[port](a5n1) [left of=a5,node distance = 5.5mm] {}; 
    \node[oport](a5x1) [right of=a5,node distance =5.5mm] {};  
  
  \node[boxloc](a6) at (4.5, -2.3) {$\begin{array}{c} A_6{:}\chkEQn{a}{b}\end{array}$};
  \node[varpass] () [below of =a6, node distance=5mm]{$(z)$};
  \node[port](a6n1) [left of=a6,node distance = 6.9mm] {}; 
   \node[port](a6x1) [right of=a6,node distance = 6.9mm] {};
   
   \node[loc] (x2) at (6,-1.2) {$ex_2$};
     \node[inv]() [below of =x2,node distance=5mm] {$\figinv{z{=}0}$};  
  
   \node[loc] (x2s) at (3,-3) {$\ddot \smile$};
     \node[inv]() [below of =x2s,node distance=3mm] {$\figinv{z{=}0}$};  
  
   \draw[trans] (n2)--(a3n1)  node [midway, above]{$\set{b}$};
   \draw[otrans] (a3x1)--(2.5,-1.9)--(2,-1.9)--(0.3,-1.9)--(0.3,-2.2)--(a4n1);
   \draw[otrans] (a3x1)--node [midway, above]{\textcolor{black}{$\set{a}$}}(a5n1);
   \draw[otrans] (a5x1)--(5.5,-1.9) -- (3.4,-1.9)--(3.4,-2.2)--(a6n1);
   \draw[otrans] (a5x1)--(x2);
   \draw[trans] (a4x1) -- (2.5,-2.3)-- (2.5,-3) -- (x2s);
   \draw[trans] (a6x1) -- (5.5,-2.3) -- (5.5,-3) -- (x2s);
   

  \draw(7, -1.5) rectangle (9,-0.5);
  \draw (8.7, -0.3) node {$\Del$};
  
  \node[loc] (n3) at (7,-1) {$en_3$};
  \node[loc] (x3) at (9,-1) {$ex_3$};
  
  \draw[trans] (n3)--(x3);


  \draw(7, -3.5) rectangle (9,-2.5);
  \draw (8.7, -2.3) node {$\chkEQn{a}{b}$};
  
  \node[loc] (n4) at (7,-3) {$en_4$};
  \node[loc] (x4) at (9,-3) {$ex_4$};
  
  \draw[trans] (n4)--(x4) node[midway,above]{$a{=}1$} node[midway,below]{$b{=}1$};


  \draw(0, -5.5) rectangle (6,-4);
  \draw (4.7, -3.8) node {$\chkDiv{b}{a}{n}:a,b\in\set{x,y}$};

  \node[loc](n5) at (0, -4.7) {$en_{5}$};  
  \node[inv]() [below of =n5,node distance=5mm] {$\figinv{z{=}0}$};  
  
  \node[boxloc](a7) at (1.5, -4.7) {$\begin{array}{c} A_7{:}M_b\end{array}$};
  \node[varpass] () [below of =a7, node distance=5mm]{$(b,z)$};
  \node[port](a7n1) [left of=a7,node distance = 6.5mm] {}; 
   \node[port](a7x1) [right of=a7,node distance =6.5mm] {};  
  
  \node (a) at (3.7,-4.7) {$\begin{array}{l} \mbox{$n{-}1$ calls to $M_b$} \\ \mbox{$(b,z)$ pass by value} \end{array}$};
  
   \node[loc] (x5) at (6,-4.7) {$ex_5$};
   \node[inv]() [below of =x5,node distance=5mm] {$\figinv{z{=}0}$};  
  
   \draw[trans] (n5)--(a7n1);
   \draw[trans,dashed] (a7x1) -- (x5);
   \draw[trans] (a) -- (x5) node[midway,above] {$a{=}n$};


  \draw(7, -5.5) rectangle (9,-4.5);
  \draw (8.7, -4.3) node {$M_b$};
  
  \node[loc] (n6) at (7,-5) {$en_6$};
  \node[loc] (x6) at (9,-5) {$ex_6$};
  
  \draw[trans] (n6)--(x6) node[midway,above]{$b{=}1$};

\end{tikzpicture}

 \caption{Games on RTA with 3 clocks : Increment $c$.}
 \label{fig_3clk_rtg_inc}
 
\end{center}
\end{figure}
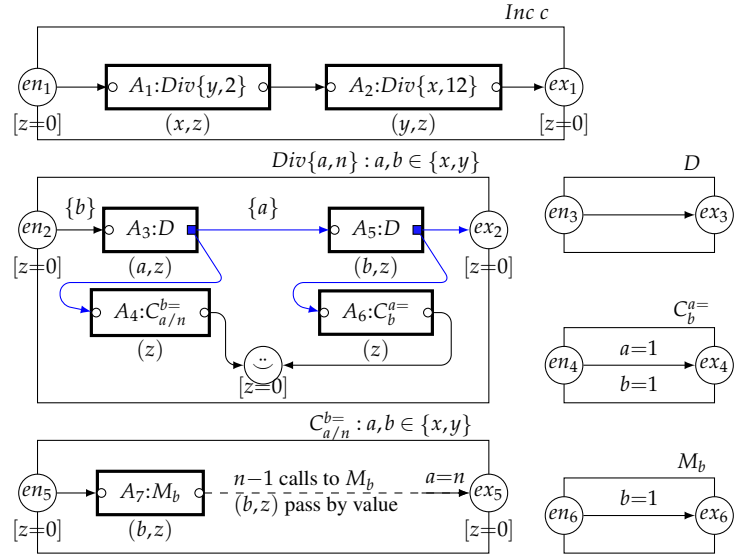

\noindent{\it Time taken}: 
Now we discuss the total time to reach a $\ddot \smile$ node   
or the exit node $ex_1$ of the component $Inc~c$ while simulating the increment instruction. 
 At the entry node $en_1$,
 clock values are $x=\frac{1}{2^{k+c}3^{k+d}}$, $y=\frac{1}{2^k}$ and $z=0$. 
 Let $\alpha = \frac{1}{2^{k+c}3^{k+d}}$ and $\beta = \frac{1}{2^{k}}$. The invaraint $z=0$ at the entry and the exit nodes $en_1$ and $ex_1$ ensures that no time elapses in these nodes and also in the return ports of $A_1$ and $A_2$. 
From the analysis above, it follows that at the return port of $A_1:\Div{y}{2}$,  $x=\alpha$, $y=\frac{\beta}{2}$ and $z=0$. Similarly at the return port of $A_2:Div\{x,12\}$, the clock values are $x=\frac{\alpha}{12}$, $y=\frac{\beta}{2}$ and $z=0$. Thus, counter $c$ has been incremented and the end of instruction $k$ has been recorded in $x$ and $y$. 
The time spent along the path from $en_1$ to $ex_1$ is the sum of times spent in $A_1:Div\{y,2\}$ and 
$A_2:Div\{x,12\}$. 
\begin{itemize}
\item Time spent in $A_1: Div\{y,2\}$. The time spent in $A_3:D$, as 
well as $A_5:D$ is both 
  $\frac{\beta}{2}$. Recall that \tort can verify that the times  $t,t'$ 
  spent in $A_3,A_5$ are both $\frac{\beta}{2}$.
If \tort enters $A_4$ to verify $t=\frac{\beta}{2}$, then the time taken is 
$2(1-t)$. In this case, the time taken to reach $\ddot \smile$  from the return port of $A_4$ is 
$t+ 2(1-t)=2-\frac{\beta}{2}$. Likewise, if \tort continued from $A_3$ to $A_5$, and goes on to verify that 
the time $t'$ spent in $A_5$ is also $\frac{\beta}{2}$, 
then the total time spent before reaching the $\ddot \smile$ from 
the return port of $A_6$ is 
  $t+t'+(1-t')=1+t=1+\frac{\beta}{2}$. Thus, 
if we are back at the return port of $A_1$, the time spent in $A_1$ is $t+t'=\beta$. 

\item Time spent in $A_2: Div\{x,12\}$. Here, the time spent in $A_3:D$ as well as $A_5:D$ is 
$\frac{\alpha}{12}$. In case \tort verifies that the time $t$ spent in $A_3:D$ is indeed $\frac{\alpha}{12}$,
then he invokes $A_4$. The time elapsed in $C^{y=}_{x/12}$ is $12(1-t)=12(1-\frac{\alpha}{12}) < 12$. 
Likewise, if \tort continued from $A_3$ to $A_5$, and goes on to verify that 
the time $t'$ spent in $A_5$ is also $\frac{\alpha}{12}$, 
then the total time spent before reaching the $\ddot \smile$ from 
the return port of $A_6$ is 
  $t+t'+(1-t')=1+t=1+\frac{\alpha}{12}$. Thus, 
if we are back at the return port of $A_2$, the time spent in $A_2$ is $t+t'=\frac{2\alpha}{12}$. 

\item In general, the component $Div\{a,n\}$ divides the value in clock $a$ by $n$. 
 If $a=\zeta$ on entering $\Div{a}{n}$, then upon exit, its value is $a=\frac{\zeta}{n}$.  The time taken to reach the exit $ex_2$ is $2 * (\frac{\zeta}{n})$. The time taken to reach the node $\ddot \smile$ in $\Div{a}{n}$ is $<n$ (due to $n$ calls to $M_b$ component).

\item Total time spent in $Inc~c$. Thus, if we come back to the return port of $A_2$, the total time spent is 
$\beta+\frac{2\alpha}{12} < 2\beta$, on entering with $y=\beta$. Recall that $x=\alpha=\frac{1}{2^{k+c}3^{k+d}}$ and $y=\beta=\frac{1}{2^k}$ and thus $\alpha \leq \beta$ always.
 \end{itemize}

In the zero check instruction, starting with $x=\frac{1}{2^{k+c}3^{k+d}}=\alpha, y=\frac{1}{2^k}=\beta$ and $z=0$, 
we divide $x$ by 6 and $y$ by 2, to record the $(k+1)$th instruction.  
\cite{KMT14} gives the details. As in the case of the increment instruction, we show that 
the main module for simulation of a zero check instruction also takes a time $<2 \beta$, on entering with $y=\beta$. 
The module for  the decrement instruction for counter $c$ only differs from the $Increment~c$ module 
of Figure \ref{fig_3clk_rtg_inc}, in that the call to $\Div{x}{12}$ is replaced by $\Div{x}{3}$ thus updating $x$ from $\frac{1}{2^{k+c}3^{k+d}}$ to $\frac{1}{2^{k+c}3^{k+d+1}}$.  Similar is the case of incrementing and decrementing 
counter $d$.

 We obtain the full RTA simulating the two counter machine 
by connecting the entry and exit of main components of instructions according to the machine's sequence of instructions. 
If the machine halts, then the RTA has an exit node corresponding to $HALT$.
Establising that for the  $k$th instruction, the time elapsed is no more than $2 \beta$, 
 for $\beta=\frac{1}{2^k}$, we have that for the first instruction, the time
 elapsed is at most $2$, for the second instruction it is   $\frac{2}{2}$, for
 the third it is $\frac{2}{2^2}$ and so on.  
 It is straightforward to see that the total time duration is bounded from above
 by  $2 (1 + \frac{1}{2} +\frac{1}{4} + \frac{1}{8} + \frac{1}{16} + \cdots) < 4$.

We now show that the two counter machine halts iff \ach has a strategy to reach $HALT$ or $\ddot \smile$. 
Suppose the machine halts. Then the strategy for \ach is to choose the appropriate delays to update the counters in each main component. Now if \tort does not verify (by entering check components) in any of the main components, then the exit $ex_1$ of the main component is reached. If \tort decides to verify then the node $\ddot \smile$ is reached. Thus, if \ach simulates the machine correctly then either the $HALT$ exit or $\ddot \smile$  is reached  if the machine halts.    
  
Conversely, assume that the two counter machine does not halt. Then we show that
\ach has no strategy to reach either $HALT$ or $\ddot \smile$. Consider a
strategy of \ach which correctly simulates all the instructions. Then $\ddot
\smile$ is reached only if \tort chooses to verify. But if \tort does not choose
to verify then $\ddot \smile$ can not be reached. The simulation continues and
as the machine does not halt, the exit node $HALT$ is never reached. Now,
consider any other strategy of \ach which does an error in simulation (in a hope
to reach $HALT$). \tort could verify this, and in this case, the node  $\ddot
\smile$ will not be reached as the delays are incorrect. Thus \ach can not
ensure reaching $HALT$ or $\ddot \smile$  with a simulation error.  
This is because there exists a strategy of \tort which can check the error
 and \ach thus can not win irrespective of all strategies of \tort.
\end{proof}
\subsection{Time Bounded Reachability Games in RSA}
\label{undec:3ws-GF}

\begin{lemma}
\label{4sw-GF-TB}
The time bounded reachability game problem is undecidable for glitch-free recursive stopwatch automata 
  with at least 4 stopwatches.
\end{lemma}
\begin{proof}
We outline quickly the changes as compared to Lemma \ref{3clk} for the case of the increment instruction. 
 Figure \ref{fig_3sw_rsg_inc} gives the component for incrementing 
counter $c$. There are 4 stopwatches $x,y,z,u$. The encoding 
of the counters in the variables is similar to Lemma \ref{3clk}: at the entry 
node of each main component simulating the $k$th instruction, we have 
$x=\frac{1}{2^{c+k}3^{d+k}}=\alpha$, $y=\frac{1}{2^k}=\beta$ and $z=0$, where $c,d$ are the current values of the counters. We use the extra stopwatch $u$ for rough work and hence we do not ensure that $u=0$ when a component is entered.

\begin{figure}
\begin{center}
 \begin{tikzpicture}[node distance=4cm] 
\tikzstyle{every node}=[font=\scriptsize]
\draw(0, 0) rectangle (7,1.5);
\draw (6.5, 1.7) node {$\incc$};

  \node[loc](n1) at (0, 0.7) {$en_{1}$};  
  \node[inv]() [below of =n1, node distance=5mm]{$\figinv{z{=}0}$};
  \node[rate]() [above of =n1,node distance=5mm] {$\figrate{z}$};  
  
  \node[boxloc](a1) at (2, 0.7) {$\begin{array}{c} A_1{:}\Div{y}{2}\end{array}$};
  \node[port](a1n1) [left of=a1,node distance = 10mm] {}; 
  \node[port](a1x1) [right of=a1,node distance = 10mm] {};
  \node[rate]() [above of =a1x1,node distance=5mm] {$\figrate{z}$};  
  
  \node[boxloc](a2) at (5, 0.7) {$\begin{array}{c} A_2{:}\Div{x}{12}\end{array}$};
  \node[port](a2n1) [left of=a2,node distance = 11mm] {}; 
  \node[port](a2x1) [right of=a2,node distance = 11mm] {};
  \node[rate]() [above of =a2x1,node distance=5mm] {$\figrate{z}$};  
  
  \node[loc](x1) at (7, 0.7) {$ex_{1}$};
  \node[inv]() [below of =x1, node distance=5mm]{$\figinv{z{=}0}$};
  
  \draw[trans] (n1)--(a1n1);
  \draw[trans] (a1x1) -- (a2n1);
  \draw[trans] (a2x1) -- (x1);

  \draw(0, -3.5) rectangle (6,-0.5);
  \draw (4.5, -0.3) node {$\Div{a}{n}:a\in\set{x,y}$};

  \node[loc](n2) at (0, -1.2) {$en_{2}$};  
  \node[inv]() [below of =n2,node distance=5mm] {$\figinv{z{=}0}$};  
  \node[rate]() [above of =n2,node distance=5mm] {$\figrate{z}$};  
  
  \node[loc](l1) at (1, -1.2) {$l_1$};    
  \node[rate]() [above of = l1, node distance=5mm] {$\figrate{u}$};
  
  \node[oloc](l2) at (2, -1.2) {$l_2$}; 
  \node[rate]() [above of =l2,node distance=5mm] {$\figrate{z}$};  
  
  \node[boxloc](a3) at (1.5, -2.3) {$\begin{array}{c} A_3{:}\chkDiv{u}{a}{n}\end{array}$};

  \node[port](a3n1) [left of=a3,node distance = 7.2mm] {}; 
  \node[port](a3x1) [right of=a3,node distance = 7.2mm] {}; 
  \node[rate]() [above of =a3x1,node distance=5mm] {$\figrate{z}$};  
  
  \node[loc](l3) at (4, -1.2) {$l_3$};  
  \node[rate]() [above of =l3,node distance=5mm] {$\figrate{a}$};  
  
  \node[oloc](l4) at (5, -1.2) {$l_4$}; 
  \node[rate]() [above of =l4,node distance=5mm] {$\figrate{z}$};  
  
  \node[boxloc](a4) at (4.5, -2.3) {$\begin{array}{c} A_4{:}\chkEQn{a}{u}\end{array}$};
  \node[port](a4n1) [left of=a4,node distance = 6.9mm] {}; 
   \node[port](a4x1) [right of=a4,node distance = 6.9mm] {};
   \node[rate]() [above of =a4x1,node distance=5mm] {$\figrate{z}$};  
   
   \node[loc] (x2) at (6,-1.2) {$ex_2$};
     \node[inv]() [below of =x2,node distance=5mm] {$\figinv{z{=}0}$};  
  
   \node[loc] (x2s) at (3,-3) {$\ddot \smile$};
     \node[inv]() [below of =x2s,node distance=3mm] {$\figinv{z{=}0}$};  
     \node[rate]() [above of =x2s,node distance=5mm] {$\figrate{z}$};  
  
   \draw[trans] (n2)--(l1)  node [midway, above]{$\set{u}$};
   \draw[trans] (l1) -- (l2);
   \draw[otrans] (l2)--(2,-1.9)--(0.3,-1.9)--(0.3,-2.2)--(a3n1);
   \draw[otrans] (l2)--node [midway, above]{\textcolor{black}{$\set{a}$}}(l3);
   \draw[trans] (l3) -- (l4);
   \draw[otrans] (l4)--(5,-1.9) -- (3.4,-1.9)--(3.4,-2.2)--(a4n1);
   \draw[otrans] (l4)--(x2);
   \draw[trans] (a3x1) -- (2.5,-2.3)-- (2.5,-3) -- (x2s);
   \draw[trans] (a4x1) -- (5.5,-2.3) -- (5.5,-3) -- (x2s);


  \draw(7, -3) rectangle (9,-1.5);
  \draw (8.7, -1.3) node {$\chkEQn{a}{u}$};
  
  \node[loc] (n4) at (7,-2.2) {$en_4$};
  \node[rate]() [above of =n4,node distance=5mm] {$\figrate{a,u}$};

  \node[loc] (x4) at (9,-2.2) {$ex_4$};
  
  \draw[trans] (n4)--(x4) node[midway,above]{$a{=}1$} node[midway,below]{$u{=}1$};


  \draw(0, -5.5) rectangle (6,-4);
  \draw (4.7, -3.8) node {$\chkDiv{u}{a}{n}:a,b\in\set{x,y}$};

  \node[loc](n3) at (0, -4.7) {$en_{3}$};  
  \node[inv]() [below of =n3,node distance=5mm] {$\figinv{z{=}0}$};  
  \node[rate]() [above of =n3,node distance=5mm] {$\figrate{z}$};    
  
  \node[boxloc](a5) at (1.5, -4.7) {$\begin{array}{c} A_5{:}M_u\end{array}$};
  \node[port](a5n1) [left of=a5,node distance = 6.5mm] {}; 
   \node[port](a5x1) [right of=a5,node distance =6.5mm] {};  
    \node[rate]() [above of =a5x1,node distance=5mm] {$\figrate{z}$};    
  
  \node (a) at (3.7,-4.7) {$\begin{array}{l} \mbox{$n{-}1$ calls to $M_u$} \\ \mbox{   } \end{array}$};
  
   \node[loc] (x3) at (6,-4.7) {$ex_3$};
   \node[inv]() [below of =x3,node distance=5mm] {$\figinv{z{=}0}$};  
  
   \draw[trans] (n3)--(a5n1)node[midway,above] {$\set{b}$};
   \draw[trans,dashed] (a5x1) -- (x3);
   \draw[trans] (a) -- (x3) node[midway,above] {$a{=}n$};


  \draw(7, -5.5) rectangle (10,-4);
  \draw (8.7, -3.7) node {$M_u:a,b\in\set{x,y}$};
  
  \node[loc] (n5) at (7,-4.7) {$en_5$};
  \node[rate]() [above of =n5,node distance=5mm] {$\figrate{a,b,u}$};  
  
  \node[loc] (l) at ( 8.5,-4.7) {$l$};
  \node[rate]() [above of =l,node distance=5mm] {$\figrate{b,u}$};  
  
  \node[loc] (x5) at (10,-4.7) {$ex_5$};
  
  \draw[trans] (n5)--(l) node[midway,above]{$u{=}1$}node[midway,below]{$\set{u}$};
  \draw[trans] (l)--(x5) node[midway,above]{$b{=}1$}node[midway,below]{$\set{b}$};

\end{tikzpicture}

 \caption{Games on Glitchfree-RSA with 4 stopwatches : Increment $c$. Note that the variables that tick in a location are indicated above it. Due to semantics of RSA, no time elapses in the call ports and exit nodes and hence variable-ticking is not mentioned for these locations.}
 \label{fig_3sw_rsg_inc}
 
\end{center}
\end{figure}
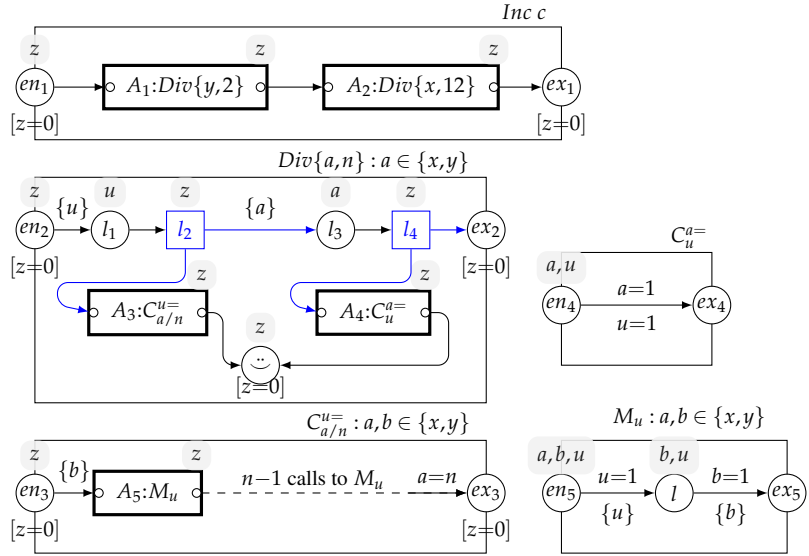

As was the case in Lemma \ref{3clk}, simulation of the $(k+1)$th instruction, incrementing $c$ amounts to 
dividing $y$ by 2 and $x$ by 12. 
 In Lemma \ref{3clk}, it was possible to pass some clocks by value, and some by reference, but 
 here, all variables must be either passed by value or by reference. 
  The $Div\{a,n\}$ module here is similar to that in Lemma \ref{3clk}: 
  the box $[A_3:D]$ in Figure \ref{fig_3clk_rtg_inc} is replaced by the node $l_1$, where only $u$ ticks and accumulates a time $t$. (Recall that $u$ is the stopwatch used for rough work and has no bearing on the encoding.) In node $l_2$, 
  only $z$ ticks.
  $l_2$ is a node belonging to \tort. 
  The time $t$ spent at $l_1$ must be exactly $t=\frac{\beta}{2}$, where $\beta=\frac{1}{2^k}$ 
  is the value of $a=y$ on entering $Div\{y,2\}$. When $t=\frac{\beta}{2}$,
  \ach can reach $\ddot \smile$ even when 
  \tort    enters the check module $C^{u=}_{y/2}$.  
    Again, note that 
  the module $C^{u=}_{y/2}$ is similar to the one in Figure  \ref{fig_3clk_rtg_inc}. We use the clock $x$ 
  (instantiating $b=x$) for rough work in this component. Due to this, the earlier value of $x$ is lost. However, this does not affect the machine simulation as we reach the node $\ddot \smile$, and the simulation does not continue.
   $C^{u=}_{y/2}$ calls the component $M_u$: 
   at the entry node $en_5$ (of $M_u$), we have $b=x=0$, $u=t$ and $a=y=\beta$. $a,b,u$ tick at $en_5$. 
   A time $1-t$ is spent at $en_5$, obtaining $b=1-t,u=0,a=\beta+(1-t)$ at $l$. 
   At $l$, only $b,u$ tick obtaining $b=0,u=t,a=\beta+(1-t)$ at $ex_5$. 
   A second invocation of 
    $C^{u=}_{y/2}$ gives $b=0,u=t,a=\beta+2(1-t)$.
    To reach $ex_3$, $a$ must be exactly 2; we thus need $t=\frac{\beta}{2}$. The time elapsed 
    in one invocation of $M_u$ is 
1 time unit; thus a total of 2+t time units is elapsed  before reaching $\ddot \smile$ (via module $C^{u=}_{y/2}$in $\Div{a}{n}$).
   If \tort skips the check at 
    $l_2$ and proceeds to $l_3$ resetting $a$(i.e; $y$), we
      have at $l_3$, $z=0,u=t=\frac{\beta}{2}$ and $a=0$.
      Only $a$ ticks at $l_3$, $a$ is supposed to ``catch up'' with $u$ at $l_3$, by elapsing $t=\frac{\beta}{2}$ 
      in $l_3$. Again, at $l_4$, only $z$ ticks. \tort can verify whether $a=u$ 
      by going to $C_{u}^{a=}$. The component 
$C_{u}^{a=}$ is exactly same as that in Figure 
\ref{fig_3clk_rtg_inc}. 
A time of $1-t$ is elapsed in 
$C_{u}^{a=}$. Thus, the time taken to reach $\ddot \smile$ 
from $C_{u}^{a=}$ is $t+t+1-t=1+t$. Thus, 
the exit node $ex_2$ of $Div\{a,n\}$ is reached in time $2t=2\frac{\beta}{2}=\beta$. 
As was the case in Lemma \ref{3clk}, the time taken to reach the exit node 
of $Inc~c$, starting with $y=\beta,x=\alpha,z=0$ is 
$\beta+ 2 \frac{\alpha}{12} < 2\beta$. Also, the time taken by $Div\{a,n\}$ 
on entering with $a=\zeta$ is $2\frac{\zeta}{n}$.

To summarize, the time taken to reach the exit node of the $Inc~c$ 
component is $< 2 \beta$, on entering with $y=\beta$. 
Also,  the component $Div\{a,n\}$ divides the value in clock $a$ by $n$. 
 If $a=\zeta$ on entering $\Div{a}{n}$, then upon exit, its value is $a=\frac{\zeta}{n}$.  The time taken to reach the exit $ex_2$ is $2 * (\frac{\zeta}{n})$. The time taken to reach the node $\ddot \smile$ in $\Div{a}{n}$ is $< n+1$ (due to $n$ calls to $M_u$ component). 

The component for zero check instruction, as in Lemma \ref{3clk}, divides $x$ by 6 and  
$y$ by 2; similarly, the component for decrement $c$ instruction 
divides  $y$ by 2 and $x$ by 3. The time to reach the exit node of any component corresponding to an instruction   
 is $<2 \beta$, on enetering the component with $y=\beta$. The total time taken
  for the simulation of the two counter machine here also, is $<4$. 
    Details can be found in \cite{KMT14}. 
\end{proof}

\noindent
The following lemma is an easy corollary of Lemma~\ref{3clk}.
\begin{lemma}
 The time bounded reachability game problem is undecidable for unrestricted
 recursive stopwatch automata with at least 3 stopwatches.  
\end{lemma}

%
%
%

\section{Conclusion}
The main result of this paper is that time-bounded reachability game problem for
recursive timed automata is undecidable for automata with three or more clocks. 
We also showed that for recursive stopwatch automata the reachability problem turns
undecidable even for glitch-free variant with $4$ stopwatches, and the
corresponding time-bounded problem is undecidable for automata with $3$
stopwatches.  
The decidability of time-bounded reachability game for recursive timed automata
with $2$ clocks is an open problem.

\bibliographystyle{eptcs}
\bibliography{papers}


\end{document}